\documentclass{aptpub}

\RequirePackage{graphicx}
\usepackage{epstopdf}
\authornames{D. ARISTOFF} 
\shorttitle{Percolation of hard disks} 

\begin{document}

\title{Percolation of hard disks} 

\authorone[University of Minnesota]{D. Aristoff} 

\addressone{Department of Mathematics, 127 Vincent Hall, 206 Church St. SE, Minneapolis, MN 55455} 

\begin{abstract}
Random arrangements of points in the plane, interacting only 
through a simple hard core exclusion, are considered. An intensity 
parameter controls the average density of arrangements, in analogy 
with the Poisson point process. It is proved that at high intensity, 
an infinite connected cluster of excluded volume appears almost surely.
\end{abstract}



\section{Introduction}

Consider a random arrangement of points in the plane. 
Suppose that each pair of points at distance less than $L$ 
from one another are joined by an edge, and let $G$ be the 
resulting graph. An important question in percolation theory is: 
Does $G$ have an infinite connected component? 

A key problem in answering this question is in defining 
what is meant by a random arrangement of points. A standard model is the {\it Poisson 
point process}, in which the probability that a (Borel)
set $A$ contains $k$ points of the random arrangement 
is Poisson distributed with parameter $\lambda |A|$, 
where $|\cdot|$ is Lebesgue measure and $\lambda$ is the {\it intensity} 
of the process. Events in disjoint sets are independent \cite{Daley}. 
Here $\lambda$ is the (average) density of 
arrangements of points; it can be shown that if $\lambda$ is 
greater than some critical value $\lambda_c$, then $G$ has an 
infinite connected component with probability one \cite{Meester}. 
(Of course $\lambda_c$ depends on the connection distance $L$.)

The Poisson point process is closely related to the 
(grand canonical) {\it Gibbs distribution} 
of statistical mechanics (with particle interaction set to zero and momentum variables integrated out) 
in the sense that they give nearly identical 
probabilistic descriptions of arrangements of points in large finite subsets of the plane.  
The Gibbs distributions, however, also allow for interactions among the 
points. Suppose the points interact through a simple exclusion of radius $2r >0$. 
(That is, each pair of points is separated by a distance of at least $2r$.) Each 
arrangement of points can then be imagined as a collection of {\it hard core} 
(i.e., nonoverlapping) disks of radius $r$.

There is a Gibbs distribution on arrangements of points with exclusion radius $2r$ in finite subsets 
of the plane which, like the Poisson process, gives equal probabilistic weight to every arrangement of the same density. 
Furthermore a probability measure can be defined 
on such arrangements in the whole plane, such that in a certain sense its restriction to finite 
subsets has the Gibbs distribution. This probability measure, called an (infinite volume)
{\it Gibbs measure}, has been extensively studied (see e.g. \cite{Ruelle1},\cite{Ruelle2},\cite{Georgii}). 

It is natural to ask whether $G$ has an infinite connected component 
when the points in $G$ are sampled from a Gibbs measure with an 
exclusion of radius $2r$. If $r << L$, one can 
argue that the exclusion is insignificant and that, 
by analogy with the Poisson process, there is some critical 
{\it activity}, $z_c$, such that $G$ almost surely has an infinite connected component 
for $z>z_c$. (See Section 7 of \cite{Radin} for a sketch of a proof in 
this direction.) Here the activity $z$ is a parameter analogous to the intensity 
of the Poisson process. 

If $r$ and $L$ are close the qualitative 
relationship with the Poisson point process 
is less clear, at least as it pertains to percolation. In particular, let $L < 4r$. 
Then the percolation question is closely related to {\it excluded volume}. 
(The excluded volume corresponding to an arrangement of points is the set 
of all points which, due to the exclusion radius, cannot be added to the arrangement.) 
If $G$ has an infinite component for such $L$, then there is an infinite 
connected region of excluded volume. The latter event has been associated 
with the gas/liquid phase transition in equilibrium statistical mechanics  
\cite{Kratky},\cite{Woodcock}. Below it is proved that given $L > 3r$, 
with points distributed under a Gibbs measure with an exclusion of radius $2r$, $G$  
has an infinite connected component almost surely whenever the activity $z$ is sufficiently large.

Little is known about qualitative properties of typical samples 
from a Gibbs measure (with exclusion) 
when $z$ is large; even simulations have been inconclusive, although a recent 
large-scale study \cite{Krauth} may settle some questions. It is expected (but not 
proven) that when $z$ is large, typical arrangements exhibit long-range orientational 
order \cite{Krauth}. On the other hand, it has been shown that there can be no long-range 
positional order at any $z$ (see \cite{Richthammer}; this is an extension 
of the famous Mermin-Wagner theorem to the case of hard core interactions). 
The absence of long-range positional order makes the percolation question 
even more pertinent.

\section{Notation, probability measure, and sketch of proof}

Fix $r > 0$, and define 
\begin{equation*}
{\Omega} = \{\omega \subset {\mathbb R}^2\,:\, |x-y|\ge 2r \,\,\forall\,\, x\ne y \in \omega\} \subset {\mathcal P}({\mathbb R}^2).
\end{equation*} 
In particular $\emptyset \in \Omega$. (Here ${\mathcal P}({\mathbb R}^2)$ is the set of subsets of ${\mathbb R}^2$.) 
Let $\mathcal T$ be the topology on $\Omega$ generated by the subbasis of sets of the form 
$\{\omega \in {\Omega}\,:\, \#(\omega \cap U) = \#(\omega\cap K) = m\}$
for compact sets $K \subset {\mathbb R}^2$, 
open sets $U \subset K$, and positive integers $m$. Here $\#\zeta$ is the number of elements in the set 
$\zeta$. Let ${\mathcal F}$ be the $\sigma$-algebra of Borel sets with respect to the topology $\mathcal T$; 
it is known that $\mathcal F$ is generated by sets of the form $\{\omega \in {\Omega}\,:\, \#(\omega \cap B) = m\}$
for bounded Borel sets $B \subset {\mathbb R}^2$ 
and nonnegative integers $m$ \cite{Ruelle2}. 
Let $\Lambda_n = [-n,n]^2\subset{\mathbb R}^2$, and given $A \in {\mathcal F}$, define 
\begin{align*}
&A_{n,N} = \{(x_1,\ldots,x_N)\,:\,\{x_1,\ldots,x_N\} \in A, \, \{x_1,\ldots,x_N\}\subset \Lambda_n\} \subset ({\mathbb R}^2)^N\\
&L_{n,N}(A) = \frac{1}{N!} \int_{A_{n,N}} dx_1\ldots dx_N\\
&L_{n,z}(A) = \sum_{N=1}^\infty z^N\,L_{n,N}(A).
\end{align*} 
For $\zeta \in \Omega$ and $n\in \mathbb N$ define 
\begin{equation*}
 \Omega_{n,\zeta} = \{\omega \in \Omega\,:\, \omega \subset \Lambda_n,\, \omega \cup (\zeta\setminus \Lambda_n) \in \Omega\}.
\end{equation*}
It is easily seen that $\Omega_{n,\zeta} \in \mathcal F$. For $\zeta \in {\Omega}$, $z \in {\mathbb R}$, and $n \in {\mathbb N}$, define the 
{\bf grand canonical Gibbs distribution $G_{n,z,\zeta}$ with boundary condition $\zeta$ on $\Lambda_n$} by  
\begin{equation}\label{gibbsd}
G_{n,z,\zeta}(A) = \frac{L_{n,z}(A \cap \Omega_{n,\zeta})}{L_{n,z}(\Omega_{n,\zeta})}
\end{equation}
for $A \in \mathcal F$. The Gibbs distribution $G_{n,z,\zeta}$ is a probability measure on $({\Omega},{\mathcal F})$ with 
support in $\Omega_{n,\zeta}$. A measure $\mu_z$ on 
$({\Omega},{\mathcal F})$ is called a {\bf Gibbs measure} if $\mu_z({\Omega}) = 1$ and 
for all $n \in {\mathbb N}$ and all measurable functions $f:{\Omega} \to [0,\infty)$, 
\begin{equation}\label{gibbs}
\int_{\Omega} f(\omega)\mu_z(d\omega) = 
\int_{\Omega} \mu_z(d\zeta) \int_{\Omega_{n,\zeta}} G_{n,z,\zeta}(d\omega)\,f(\omega\cup(\zeta\setminus \Lambda_n)).
\end{equation}
It is well known that $\mu_z$ exists for every $z$. 
(For a proof of existence, see 
\cite{Ruelle2}.) However, $\mu_z$ may be non-unique. When 
$\mu_z$ is referred to below, it is assumed $\mu_z$ is an 
arbitrary Gibbs measure, unless otherwise specified.

For $s >0$, $P,Q \subset {\mathbb R}^2$ and $x \in {\mathbb R}^2$, define 
\begin{align*}
&B_s(x) = \{y \in {\mathbb R}^2\,:\,|x-y|\le s\}\\
&d(P,Q) = \inf\{|p-q|\,:\,p \in P,\,q \in Q\} \\
&P - x = \{p-x\,:\,p \in P\}
\end{align*}
and call $P$ {\bf infinite} if for every $n$, $P$ is not a subset of $\Lambda_n$. 

Let $L>3r$. The main result of this paper, Theorem~\ref{theorem3}, 
states that for $z$ sufficiently large, 
$\cup_{x \in \omega} B_{L/2}(x)$ 
has an infinite connected component $\mu_z$-almost surely, for 
all Gibbs measures $\mu_z$.
As a preliminary step the following is shown in 
Theorem~\ref{theorem2}: Let 
$A_{inf}$ be the event that $\cup_{x \in \omega} B_{L/2}(x)$ has an infinite connected 
component, $W$, such that $d(0,W)\le L/2$. 
Then $\lim_{z\to \infty}\mu_z(A_{inf}) = 1$, 
uniformly in all Gibbs measures $\mu_z$.

Here an outline of the proof of Theorem~\ref{theorem2} 
is sketched. Write $R = \delta + 3r/2$ with $\delta>0$, 
with $R$ chosen to be slightly smaller than $L/2$. 
Let $\Psi:{\mathbb R}^2 \to (\epsilon{\mathbb Z})^2$ be a discretization of space,  
with $\epsilon$ much smaller than $r$ and $\delta$. 
Let $\omega \in {\Omega}$, and suppose $\cup_{x \in \omega} B_{R}(\Psi(x))$ has a finite 
connected component $W$. The boundary of $W$ is comprised of a number of closed curves; let $\gamma$ be the one 
which encloses a region $W_\gamma$ containing all the others, 
and assume $\gamma$ is comprised of exactly $K$ arcs.  
Let $A_\gamma$ be the set of all $\omega \in \Omega$ 
for which the curve $\gamma$ arises as above.
It can be shown that there is a vector $u_0 \in {\mathbb R}^2$ of magnitude $\sim r$ and a 
map $\phi\,:\, A_\gamma \to \Omega$ defined by 
$\phi(\omega) = ((\omega \cap W_\gamma) - u_0) \cup (\omega \setminus W_\gamma)$
with the following properties: $L_{n,z}(\phi(A)) = L_{n,z}(A)$ for all 
measurable $A \subset A_\gamma$, and there exist $x_1,x_2,\ldots,x_M \in {\mathbb R}^2$, 
with $M=\lceil cK\rceil$ and $c$ a positive constant (depending only on $\delta$ and $r$, 
and not on $\gamma$), 
such that for all $\omega \in A_\gamma$ and 
$i\ne j \in \{1,2,\ldots,M\}$,  
\begin{equation*}
d(x_i,\phi(\omega)) \ge \delta/2 + 2r \hskip20pt \hbox{and} \hskip20pt |x_i-x_j|  \ge \delta +  2r.
\end{equation*} 
Then with 
$A_\gamma^\phi = \{\phi(\omega) \cup \{y_1,y_2,\ldots,y_M\}\,:\,\omega \in A_\gamma,\, y_i \in B_{\delta/2}(x_i)\}$ 
one can show that
\begin{equation*}
 G_{n,z, \zeta}(A_\gamma) \le \frac{G_{n,z,\zeta}(A_\gamma)}{G_{n,z,\zeta}(A_\gamma^\phi)} = (\pi \delta^2 z/4)^{-M}
\end{equation*}
provided $n$ is large enough. It follows that 
$\mu_z(A_\gamma) \le (\pi \delta^2 z/4)^{-M}$. 
 
Let $A_{inf}^\Psi$ be the event that $\cup_{x \in \omega} B_{R}(\Psi(x))$ 
has an infinite connected component $W$ such that $d(0,W)\le r/2$. 
Consider only those {\it finite} connected components 
$W$ of $\cup_{x \in \omega} B_{R}(\Psi(x))$ 
such that $d(0,W) \le r/2$. A counting argument shows that the 
number of curves $\gamma$ with $K$ arcs corresponding to such $W$ is bounded above by 
\begin{equation*}
\left(\frac{(K+1)H}{\epsilon}\right)^2 \left(\frac{H}{\epsilon}\right)^{2(K-1)}
\end{equation*}
where $H$ depends only on $\delta$ and $r$. 
So the $\mu_z$-probability that there is 
a finite connected component $W$ of $\cup_{x \in \omega} B_{R}(\Psi(x))$ 
such that $d(0,W) \le r/2$ is less than 
\begin{equation*}
\sum_{K=1}^{\infty} \left(\frac{(K+1)H}{\epsilon}\right)^2 \left(\frac{H}{\epsilon}\right)^{2(K-1)} \left(\frac{\pi\delta^2 z}{4}\right)^{-\lceil cK\rceil}.
\end{equation*}
This summation approaches zero as $z\to \infty$. A simpler version of the 
above arguments shows that the $\mu_z$-probability that 
$d(0,W)>r/2$ for {\it all} connected components $W$ of $\cup_{x \in \omega} B_{R}(\Psi(x))$
also approaches zero as $z \to \infty$. It follows that 
$\lim_{z\to \infty}\mu_z(A_{inf}^\Psi) = 1$.
The continuous space corollary is the statement $\lim_{z\to \infty} \mu_z(A_{inf}) = 1$, which 
is deduced by an appropriate choice of $R$; since all of the above estimates 
apply to arbitrary Gibbs measures $\mu_z$, the convergence is uniform in $\mu_z$.

\section{Discretization and contours}

\begin{figure}
\begin{center}
\includegraphics[scale=0.45]{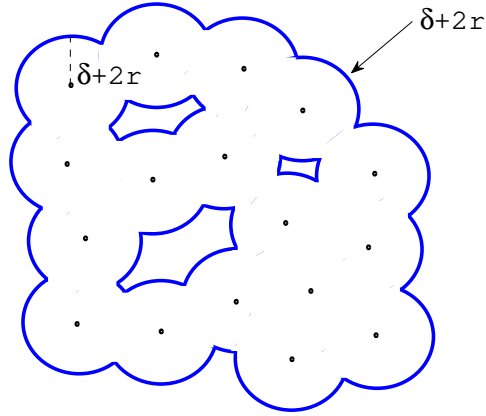}
\end{center}
\vskip-20pt
\caption{The outer curve is a contour $\gamma= \gamma_{\omega,\omega'}$ of size $13$. All the points pictured belong to $\Psi(\omega')$. }\label{figure0}
\end{figure}

Throughout $R$, $\delta$ and $\epsilon$ are fixed with 
$R = \delta+3r/2$, $\delta \in (0,r/2)$ 
and $\epsilon \in (0,\delta/2)$. Define 
$\Psi:{\mathbb R}^2 \to (\epsilon{\mathbb Z})^2$ as follows: 
for $n, m \in {\mathbb Z}$, if 
\begin{equation*}
(x,y) \in [\epsilon m-\epsilon/2,\epsilon m + \epsilon/2)\times 
[\epsilon n -\epsilon/2,\epsilon n + \epsilon/2)
\end{equation*} 
then set 
\begin{equation*}
\Psi(x,y) = (\epsilon m, \epsilon n).
\end{equation*}
Note that $|\Psi(x) -x| < \epsilon$ for all $x \in {\mathbb R}^2$. 
Furthermore $\Psi$ is Borel measurable in the sense that 
$\Psi^{-1}(P)$ is a Borel set for any $P \subset (\epsilon{\mathbb Z})^2$. 
(The dependence of $\Psi$ on $\epsilon$ will be suppressed.)

Let $\omega \in {\Omega}$. The connected components of 
$\cup_{x \in \omega} B_R(\Psi(x))$ naturally partition $\omega$ into 
subsets $\omega' \subset \omega$; each $\omega'$ consists exactly of all the 
points $x \in \omega$ such that $\Psi(x)$ belongs to a given 
connected component of 
$\cup_{x \in \omega} B_R(\Psi(x))$. The subsets $\omega'$ will be called 
{\bf components} of $\omega$. A component $\omega'$ of $\omega$ 
is said to be {\bf finite} if $\omega' \subset \Lambda_n$ for some $n$. For 
each finite component $\omega'$ of $\omega \in \Omega$, consider the set 
$W_{\omega,\omega'} = \cup_{x \in \omega'} B_{\delta+2r}(\Psi(x))$.
Since $\delta + 2r \ge R$, $W_{\omega,\omega'}$ is connected. 
(It will also be assumed throughout 
that $r, \delta \in \mathbb Q$ and that $\epsilon$ 
is transcendental. This assumption implies that 
if two disks in $W_{\omega,\omega'}$ intersect, then they 
overlap.) Consider now the boundary $\partial W_{\omega,\omega'}$ of $W_{\omega,\omega'}$. By the above, 
$\partial W_{\omega,\omega'}$ is a union of (images of) simple closed curves, one of which 
encloses a region containing all the others. Define 
$\gamma = {\gamma}_{\omega,\omega'} \subset {\mathbb R}^2$ to be the latter curve; 
$\gamma$ will be called a {\bf contour} of $\omega$. A contour $\gamma$ is (the image of) a 
simple closed curve comprised of circle arcs. The total 
number of circle arcs in $\gamma$ is called the {\bf size of the contour}. See Figure~\ref{figure0}. 
{\bf The region enclosed by $\gamma$} will be denoted $W_\gamma$. It is 
emphasized that a contour $\gamma= \gamma_{\omega,\omega'}$ is defined only when $\omega'$ is a finite component 
of some $\omega \in \Omega$.

\begin{lemma}
\label{lemma1}
There exists $c >0$ 
such that the following holds. Let $\gamma$ be 
any contour of size $K>0$, and let $A_\gamma$ be the (nonempty) set 
of all $\omega \in {\Omega}$ such that $\gamma = \gamma_{\omega,\omega'}$ for some finite 
component $\omega'$ of $\omega$. Then $A_\gamma \in \mathcal F$. 
Choose $n$ such that $\gamma \subset \Lambda_n$. 
There is a map $\phi: A_\gamma \to {\Omega}$ and $x_1,x_2,\ldots,x_M {\in \mathbb R}^2$, with  
$M = \lceil c K\rceil$, such that: 
\vskip5pt

(i) $L_{n,z}(A) = L_{n,z}(\phi(A))$ for all $z$ and ${\mathcal F}$-measurable $A \subset A_\gamma$

(ii) $|x_i-x_j| \ge \delta + 2r$ for all $i\ne j \in \{1,2,\ldots,M\}$

(iii) $d(x_i, \phi(\omega)) \ge  \delta/2 + 2r$ for all $i \in \{1,2,\ldots,M\}$ and all $\omega \in A_\gamma$.

\end{lemma}

\begin{proof}
To see that $A_\gamma \in \mathcal F$, note that $A_\gamma$ can be written as a finite intersection of sets of the form 
$\{\omega \in \Omega\,:\, \#(\omega \cap\Psi^{-1}(\{x\}))= \ell\}$, where $x \in (\epsilon{\mathbb Z})^2$ 
and $\ell \in \{0,1\}$. 

\begin{figure}
\begin{center}
\includegraphics[scale=0.45]{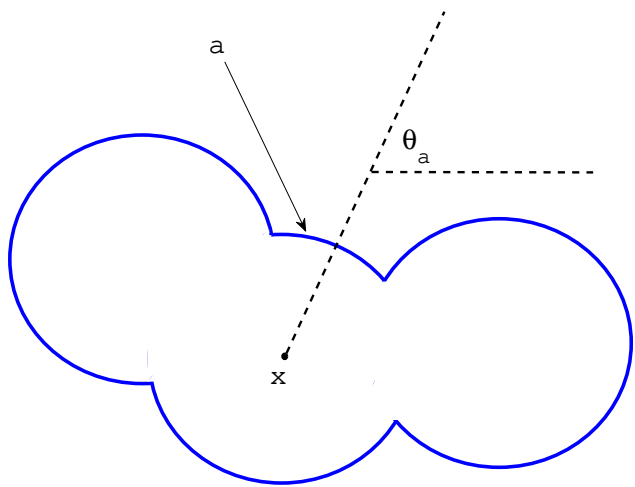}
\end{center}
\caption{A contour $\gamma_{\omega,\omega'}$ with the arc $a$. $\theta_a$ is the 
outward normal angle with respect to the midpoint of $a$. Here $x \in \Psi(\omega')$.}\label{figure1}
\end{figure}

For each circle arc $a$ of $\gamma$, let $\theta_a \in [0,2\pi)$ 
be an outward normal angle with respect to the midpoint of the arc
(see Figure~\ref{figure1}). Choose $0 < \alpha < \delta/(\delta + 2r)$ 
so that $\alpha = 2\pi/n$ for 
some $n\in {\mathbb N}$. By the pigeonhole 
principle, there is a subinterval 
$I = [v, v+\alpha) \subset [0,2\pi)$ 
such that $\lceil (2\pi)^{-1}\alpha K\rceil$ of the angles $\theta_a$ 
belong to $I$. Fix $\theta_0  \in I$ and let 
\begin{equation*}
 u_0 = \left((\delta/2 + r)\cos \theta_0, (\delta/2 + r)\sin \theta_0\right)
\end{equation*}
be the vector in the direction of $\theta_{0}$ with magnitude $\delta/2+r$. Define 
$\phi:{\mathcal P}({\mathbb R}^2) \to {\mathcal P}({\mathbb R}^2)$ by 
\begin{equation*}
\phi(X) = ((X\cap W_\gamma)-u_0) \cup (X\setminus W_\gamma).
\end{equation*}
It will be shown below that $\phi(A_\gamma) \subset \Omega$. 

Let $\omega \in A_\gamma$ be arbitrary, and let $\omega'$ be the unique 
component of $\omega$ such that $\gamma = \gamma_{\omega,\omega'}$. 
Assume $x \in \omega\setminus W_\gamma$. 
Then $d(\Psi(x),\Psi(\omega')) > 2\delta + 3r$, and so 
\begin{equation*}
d(\Psi(x),\cup_{y \in \omega'}B_{\delta + 2r}(\Psi(y)) > \delta + r.
\end{equation*}
It follows that $d(\Psi(x),\gamma) > \delta + r$, so that $d(x,\gamma) >  \delta/2 + r$. 
Now assume $x \in \omega \cap W_\gamma$. If $x \in \omega'$ then  
$d(\Psi(x),\gamma)\ge \delta + 2r$, and so $d(x,\gamma) > \delta/2 + 2r$. 
If $x \notin \omega'$ then 
\begin{equation*}
 \Psi(x) \notin \cup_{y \in \omega'}B_{2\delta + 3r}(\Psi(y))
\end{equation*}
and a simple computation shows $d(\Psi(x),\gamma) > \sqrt{5r^2 + 8r\delta + 3\delta^2} > \delta + 2r$,  
so that $d(x,\gamma) > \delta/2 + 2r$. (See Figure~\ref{figure2a}).

\begin{figure}
\begin{center}
\includegraphics[scale=0.45]{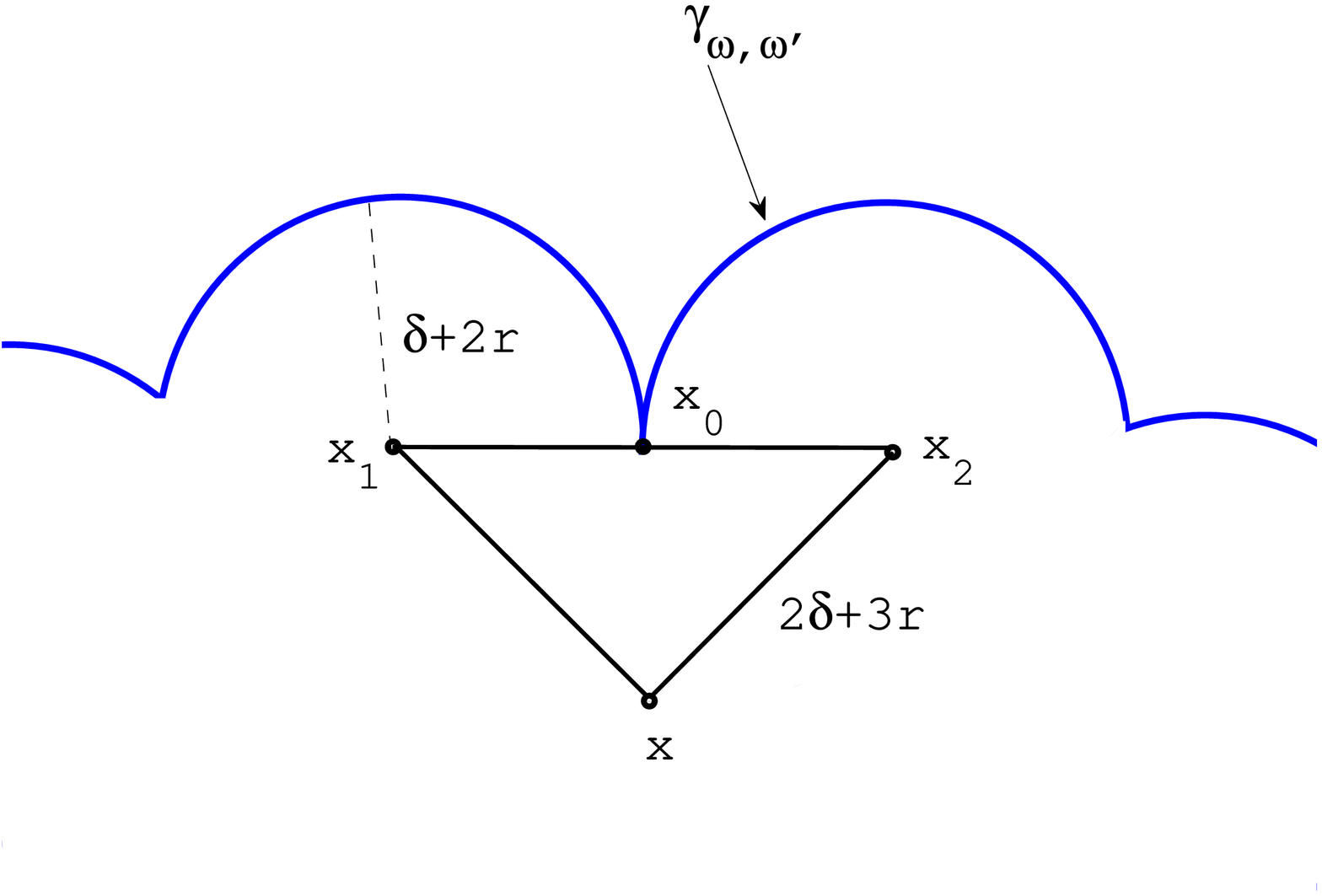}
\end{center}
\vskip-20pt
\caption{Pictured are $x_1,x_2 \in \Psi(\omega')\subset W_\gamma$, and $x \in \Psi(\omega \cap W_\gamma)$ 
but $x \notin \Psi(\omega')$. For such $x$, $d(x, \gamma) > \sqrt{5r^2+8r\delta+3\delta^2}$.  
This can be seen in the above picture, in which the distance from $x$ to $\gamma$ is minimized by placing 
$x_1$ and $x_2$ as far apart as possible.}\label{figure2a}
\end{figure}

Now let $A \subset A_\gamma$ with $A \in {\mathcal F}$, and define 
\begin{align*}
 &A^{in} = \{\omega \cap W_\gamma\,:\, \omega \in A\}  \\
 &A^{out} = \{\omega \setminus W_\gamma\,:\, \omega \in A\}.
\end{align*}
Let $\omega^{in} \in A^{in}$ and $\omega^{out} \in A^{out}$. 
By the preceding paragraph, 
\begin{align*}
&d(\omega^{out},\gamma) > \delta/2+r\\
&d(\omega^{in},\gamma) > \delta/2+2r.
\end{align*}
Let $x \in \omega^{in}$ and $y \in \omega^{out}$, and 
let $z$ be any point on the intersection 
of $\gamma$ with the line segment $\overline{xy}$. Then 
\begin{equation*}
 |x-y| = |x-z| + |y-z| > \delta/2+ 2r + \delta/2 + r = \delta + 3r.
\end{equation*}
As $|u_0| = \delta/2 + r$, it follows that 
\begin{equation*}
 |\phi(x)-\phi(y)| = |(x-u_0)-y| > \delta/2 + 2r.
\end{equation*}
By the preceding statements 
\begin{align*}
 &d(\omega^{in}, \omega^{out}) > \delta + 3r \ge 2r \\
 &d(\phi(\omega^{in}),\phi(\omega^{out})) > \delta/2 + 2r \ge 2r.
\end{align*}
In particular this shows $\phi(A) \subset \Omega$, and so $\phi(A_\gamma) \subset \Omega$. 
Also note that 
$d(\omega^{in},\gamma) > \delta/2 + 2r$ and $\gamma \subset \Lambda_n$ 
together imply $\phi(\omega^{in}) = \omega^{in} - u_0 \subset \Lambda_n$. 
Combining the above statements, 
\begin{align*}
 L_{n,N}(A) &= L_{n,N}(A^{in})\,L_{n,N}(A^{out}) \\
&= L_{n,N}(A^{in}-u_0)\,L_{n,N}(A^{out}) \\
&= L_{n,N}(\phi(A^{in}))\,L_{n,N}(\phi(A^{out})) \\
&= L_{n,N}(\phi(A)).
\end{align*}
Since $\# (\omega \cap \Lambda_n) = \#(\phi(\omega) \cap \Lambda_n)$ for each 
$\omega \in A_\gamma$, it follows that $L_{n,z}(A) = L_{n,z}(\phi(A))$. This proves (i).

\begin{figure}
\begin{center}
\includegraphics[scale=0.45]{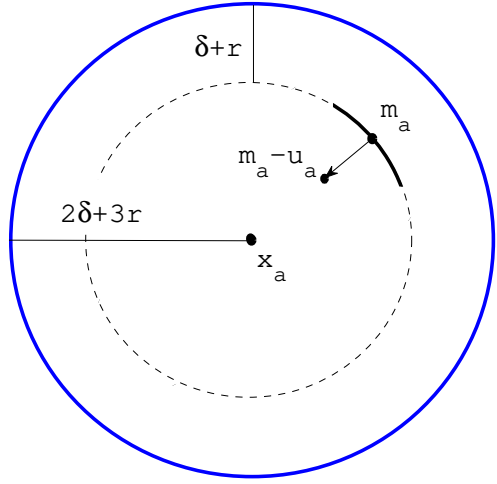}
\end{center}
\caption{The midpoint $m_a$ of the arc $a$ with corresponding normal vector $u_a$. Here $x_a \in \Psi(\omega')$. No points 
in $\Psi(\omega \setminus W_\gamma)$ can be inside the large circle. The magnitude of $u_a$ is $r+\delta/2$, and so 
the $d$-distance between $m_a-u_a$ and 
the large circle is $3\delta/2 + 2r$.}\label{figure2b}
\end{figure}

Consider now (ii) and (iii). Again let $\omega \in A_\gamma$, and let $\omega'$ 
be the unique component of $\omega$ such that $\gamma = \gamma_{\omega,\omega'}$. 
Let $a$ be an arc of $\gamma$ such that $\theta_a \in I$. Let $m_a$ be the 
midpoint of the arc, $x_a$ the center of the circle (of radius $\delta + 2r$) 
which forms the arc, and $u_a$ the vector in the direction of $\theta_a$ with 
magnitude $\delta/2 + r$. 

As $x_a \in \Psi(\omega')$, 
no points of $\Psi(\omega \setminus W_\gamma)$ are in 
$B_{2\delta+3r}(x_a)$. As $|u_a| = \delta/2 + r$, it follows that 
for any $x \in \omega \setminus W_\gamma$, 
$|\Psi(x)-(m_a-u_a)|>3\delta/2 + 2r$. (See Figure~\ref{figure2b}.) 
So for each $x \in \omega \setminus W_\gamma$, 
\begin{align*}
|\Psi(x)-(m_{a}-u_0)| &\ge |\Psi(x)-(m_{a}-u_a)| - |u_a - u_0| \\
&> \frac{3\delta}{2} + 2r - \left(\frac{\delta}{2} + r\right)\alpha\\
&> \delta + 2r,
\end{align*} 
where the last inequality comes by choice of $\alpha$. Therefore 
if $x \in \omega \setminus W_\gamma$ then 
\begin{equation*}
 |\phi(x)-(m_a-u_0)|=|x-(m_a-u_0)| >  \delta/2+2r.
\end{equation*}
On the other hand if $x\in \omega \cap W_\gamma$ then $d(\Psi(x),\gamma) \ge \delta + 2r$, 
and so 
\begin{equation*} 
|\phi(x)-(m_a-u_0)| = |x-m_a| > \delta/2 + 2r.
\end{equation*}
Combining the above statements, if $x \in \omega$ 
then $|\phi(x)-(m_a-u_0)|>\delta/2 + 2r$. 

Now note that for any $x \in \Psi(\omega')$, a disk $B_{2r+\delta}(x)$ contributes to 
no more than $6$ distinct circle arcs in $\gamma$. In turn, each circle arc corresponds 
to a unique $x \in \Psi(\omega')$ which is the center of the circle forming the arc. 
If two arc midpoints in $\gamma$ are at distance less than $\delta+2r$ from one another, 
then the corresponding $x,y \in \Psi(\omega')$ are at distance less than $3\delta + 6r$, 
so that the (unique) points in $\omega'$ which $\Psi$ maps to $x$ and $y$ are at 
distance less than $4\delta + 6r < 8r$ from each other. By 
a simple area comparison, the number of 
points $x \in \omega$ contained in a disk of radius $8r$ is bounded above by $(9r)^2/r^2 = 81$. 
The preceding shows that, given any arc midpoint $m_a$ in $\gamma$, the number of arc 
midpoints $m_{\tilde a} \ne m_a$ in $\gamma$ such that $|m_a-m_{\tilde a}| < \delta + 2r$ 
is bounded above by $J = 6\cdot 81 = 486$. So with $c = (2\pi(J+1))^{-1}\alpha$, 
there exists a subcollection 
\begin{equation*}
 \{m_1,m_2,\ldots,m_M\} \subset \{m_a\,:\, \theta_a \in I\},\hskip10pt M = \lceil cK \rceil
\end{equation*}
of arc midpoints such that $d(m_i, m_j) \ge \delta + 2r$ for all $i \ne j \in \{1,2,\ldots,M\}$. 
By taking $x_i = m_i - u_0$ for $i \in \{1,2,\ldots,M\}$, the proof is completed. 
\end{proof}

\section{Estimates}

Using Lemma~\ref{lemma1}, the $\mu_z$-probability of seeing a given contour $\gamma$ 
is shown to be exponentially small in the size, $K$, of the contour.

\begin{figure}
\begin{center}
\includegraphics[scale=0.45]{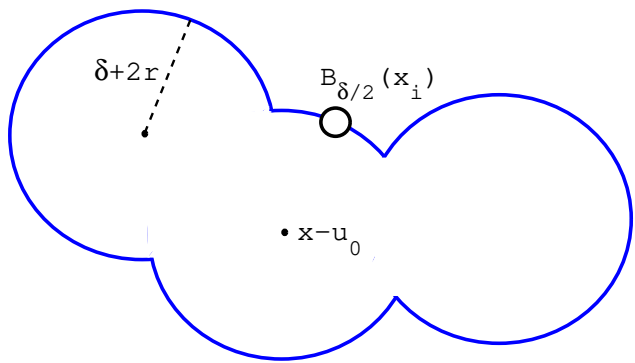}
\end{center}
\caption{A disk $B_{\delta/2}(x_i)$ centered at a midpoint of an arc of $\gamma_{\omega,\omega'} - u_0$, with 
$x \in \Psi(\omega')$.}\label{figure3}
\end{figure}

\begin{lemma}
 \label{lemma3}
There exists $c>0$ such 
that the following holds. Let $\gamma$ be any contour of size $K$, and let 
$A_\gamma$ be the set of all $\omega \in {\Omega}$ such that $\gamma = \gamma_{\omega,\omega'}$ 
for some finite component $\omega'$ of $\omega$. Then for every Gibbs measure $\mu_z$, 
\begin{equation*}
\mu_z(A_\gamma) \le (\pi \delta^2 z/4)^{-\lceil c K\rceil}.
\end{equation*}
\end{lemma}

\begin{proof}
Choose $c>0$, $\phi$ and $x_1,x_2,\ldots,x_M$ satisfying the conclusion of 
Lemma~\ref{lemma1}. Choose $\hat n$ 
so that $\gamma \subset \Lambda_{\hat n}$, and let $\zeta \in \Omega$ be arbitrary. 
For each $A \subset A_\gamma$ such that $A \in {\mathcal F}$, define 
\begin{equation*}
 {A}^\phi = \{{\omega}^\phi \subset {\mathbb R}^2 \,:\, 
{\omega}^\phi = \phi(\omega) \cup \{y_1,y_2,...,y_M\}, \, \omega \in A,\, y_i \in B_{\delta/2}(x_i)\}
\end{equation*}
(See Figure~\ref{figure3}.) By conditions (ii)-(iii) of Lemma~\ref{lemma1}, 
${A}_\gamma^\phi \subset {\Omega}$, and since $A_\gamma \in \mathcal F$ 
it is easy to see that $A_\gamma^\phi \in {\mathcal F}$.

By definition of $\phi$ and choice of $\hat n$, if $\omega \in A_\gamma$ and 
$\omega^\phi = \phi(\omega) \cup \{y_1,y_2,\ldots,y_M\}$ with $y_i \in B_{\delta/2}(x_i)$, then  
$\omega \setminus \Lambda_{{\hat n}+l} = \omega^\phi \setminus \Lambda_{{\hat n}+l}$, 
where $l = \lceil\delta+r\rceil$. 
Now let $n = {\hat n}+ l + \lceil 2r\rceil$. If $\omega \in A_\gamma$ and 
$\omega^\phi = \phi(\omega) \cup \{y_1,y_2,\ldots,y_M\}$ with $y_i \in B_{\delta/2}(x_i)$, then 
$\omega \in \Omega_{n,\zeta}$ if and only if $\omega^\phi \in \Omega_{n,\zeta}$. 
Let $A_{\gamma,n,\zeta} = A_\gamma \cap \Omega_{n,\zeta}$. 
The preceding shows that $A_{\gamma,n,\zeta}^\phi = A_\gamma^\phi \cap \Omega_{n,\zeta}$. 

Now, since each disk $B_{\delta/2}(x_i)$ has (Lebesgue) area $\pi\delta^2/4$, Lemma~\ref{lemma1} implies  
\begin{align*}
 L_{n,z}(A_{\gamma,n,\zeta}^\phi) &= (\pi \delta z/4)^M L_{n,z}(\phi(A_{\gamma,n,\zeta}))\\
&= (\pi \delta z/4)^M L_{n,z}(A_{\gamma,n,\zeta}).
\end{align*}
From definitions it is easy to see that $G_{n,z,\zeta}(A_\gamma)$ and 
$G_{n,z,\zeta}(A_\gamma^\phi)$ are positive. Thus 
\begin{align*}
{G_{n,z,\zeta}(A_\gamma)} \le \frac{G_{n,z,\zeta}(A_\gamma)}{G_{n,z,\zeta}(A_\gamma^\phi)} 
&=\frac{L_{n,z}(A_\gamma \cap \Omega_{n,\zeta})}{L_{n,z}(A_\gamma^\phi \cap \Omega_{n,\zeta})}\\
&=\frac{L_{n,z}(A_{\gamma, n,\zeta})}{L_{n,z}(A_{\gamma,n,\zeta}^\phi)}\\
&= (\pi \delta^2 z/4)^{-M}.
\end{align*}
Also by choice of $n$, if $\omega \in \Omega_{n,\zeta}$, then 
$\chi_{A_\gamma}(\omega) = \chi_{A_\gamma}(\omega \cup (\zeta\setminus \Lambda_n))$
where $\chi_{A_\gamma}: \Omega \to [0,\infty)$ is the (measurable) function 
$\chi_{A_\gamma}(\omega) = 1$ if $\omega \in A_\gamma$, and 
$\chi_{A_\gamma}(\omega) = 0$ otherwise.
Since $\zeta$ was arbitrary, 
\begin{align*}
 \mu_z(A_\gamma) &= \int_{\Omega} \mu(d\zeta) \int_{\Omega_{n,\zeta}} G_{n,z,\zeta}(d\omega)\,\chi_{A_\gamma}(\omega \cup(\zeta\setminus \Lambda_n)) \\
 &= \int_{\Omega} \mu(d\zeta) \int_{\Omega_{n,\zeta}} G_{n,z,\zeta}(d\omega) \,\chi_{A_\gamma}(\omega) \\
 &= \int_{\Omega} G_{n,z,\zeta}(A_\gamma)\, \mu(d\zeta) \\
 &\le \int_\Omega  (\pi \delta^2 z/4)^{-M}\, \mu(d\zeta) \\
&=  (\pi \delta^2 z/4)^{-M}.
\end{align*}
As $\mu_z$ was an arbitrary Gibbs measure, the proof is complete.
\end{proof}

Next an upper bound for the number of contours enclosing 
the origin is obtained:

\begin{lemma}
 \label{lemma4}
Let $\Gamma_K$ be the set of all contours $\gamma$ of size $K$ 
such that $0 \in W_\gamma$. Then
\begin{equation*}
\#{\Gamma}_K \le \left(\frac{(K+1) H}{\epsilon}\right)^2\left(\frac{H}{\epsilon}\right)^{2(K-1)}
\end{equation*}
where $H$ is a constant depending only on $r$.
\end{lemma}

\begin{proof} 
Note that each contour 
$\gamma$ is completely determined by its set of arcs, with each arc naturally corresponding 
to a unique point in $(\epsilon{\mathbb Z})^2$, namely, the center of the circle of which the arc is part.  
Let $\gamma \in \Gamma_K$. Since $\gamma$ is the (image of a) simple closed curve comprised of circle arcs, there is a sequence 
of circle arcs $a_1,a_2,\ldots,a_K$ such that $a_i$ and $a_{i+1}$ are adjacent for $i=1,2,\ldots, K-1$. 
Choose the corresponding sequence $x_1,x_2,\ldots,x_K$ of points in $(\epsilon{\mathbb Z})^2$. Then $|x_{i+1}-x_i| < 2\delta + 4r < 5r$
for $i=1,2,\ldots,K-1$. 

By a simple area comparison, the number of points in $(\epsilon {\mathbb Z})^2$ inside any disk $B_{s}(x)$ is bounded 
above by 
\begin{equation*}
\frac{\pi(s + \epsilon)^2}{\epsilon^2} < \frac{2\pi s^2}{\epsilon^2}
\end{equation*}
if $s > 3\epsilon$. As $\gamma$ encloses the origin, $x_1$ must be contained 
in a disk of radius $(K+1)5r$ around $0$. 
Therefore there are at most $2\pi [(K+1)5r]^2/\epsilon^2$ 
possibilities for $x_1$. For $i=1,2,\ldots,K-1$, 
$x_{i+1}$ must be contained in a disk of radius $5r$ around $x_i$,  
so given $x_i$ there are no more than 
$2 \pi(5r)^2/\epsilon^2$ possibilities for $x_{i+1}$. Taking 
$H = 5\sqrt{2\pi}r$, the result follows.
\end{proof}

\section{Main results}

Let $\omega \in \Omega$. If the origin is not close to an infinite component of $\omega$, 
then it is either close to a finite component of $\omega$, or it is not close to 
any component of $\omega$. The probability of the former event can be handled by 
combining Lemma~\ref{lemma3} with Lemma~\ref{lemma4}, while it is easy to control  
the probability of the latter event. This leads to the following.

\begin{theorem}
\label{theorem1}
Let $A_{inf}^\Psi$ be the set of all $\omega \in {\Omega}$ such that 
$d(0,\Psi(\omega')) \le \delta + 2r$ for some infinite component $\omega'$ of $\omega$. 
Then $A_{inf}^\Psi \in \mathcal F$, and $\lim_{z\to \infty} \mu_z(A_{inf}^\Psi) =1$, 
uniformly in all Gibbs measures $\mu_z$.
\end{theorem}

\begin{proof}
Define 
\begin{align*}
  &A_{orig} = \{\omega \in \Omega\,:\, d(0,\Psi(\omega')) > \delta + 2r  \hbox{ for all components }\omega' \hbox{ of }\omega\} \\
  &A_{fin} = \{\omega \in \Omega\,:\, d(0,\Psi(\omega')) \le \delta + 2r \hbox{ for some finite component }\omega' \hbox{ of }\omega\} \\
 &A_{cont} = \{\omega \in \Omega\,:\,0\in W_\gamma \hbox{ for some contour } \gamma = \gamma_{\omega,\omega'}\}.
\end{align*}
Note that $A_{orig}$, $A_{fin}$, and $A_{cont}$ can each be written 
as a countable union of finite intersections of sets of the form 
$\{\omega \in \Omega\,:\, \#(\omega \cap\Psi^{-1}(\{x\}))= \ell\}$ where $x \in (\epsilon{\mathbb Z})^2$ 
and $\ell \in \{0,1\}$. Thus $A_{orig}, A_{fin}, A_{cont} \in{\mathcal F}$. 

Let $A_n$ be the set of all $\omega \in \Omega$ with the following property: 
that there exist a positive integer $k$ and $x_1,x_2,\ldots,x_k \in \Psi(\omega)$ such that $|x_1|\le \delta + 2r$, 
$|x_i-x_{i+1}| \le 2R$ for $i=1,2,\ldots,k-1$, and $x_k \notin \Lambda_n$. 
Note that $A_n$ can be written as a finite union of finite intersections of sets of the form 
$\{\omega \in \Omega\,:\, \#(\omega \cap\Psi^{-1}(\{x\}))= 1\}$ where $x \in (\epsilon{\mathbb Z})^2$. 
So $A_n \in \mathcal F$. Since $A_{inf}^\Psi = \cap_{n=1}^\infty A_n$ 
it follows that $A_{inf}^\Psi \in \mathcal F$.

Note that $\Omega \setminus A_{inf}^\Psi \subset A_{orig}\cup A_{fin}$
and $A_{fin} \subset A_{cont}$, so 
\begin{align*}
 \mu_z(\Omega \setminus A_{inf}^\Psi) &\le \mu_z(A_{orig}) + \mu_z(A_{fin}) \\
&\le \mu_z(A_{orig}) + \mu_z(A_{cont}).
\end{align*}
Choose $c>0$ such that the conclusion of Lemma~\ref{lemma3} holds, 
and choose $H$ such that the conclusion of Lemma~\ref{lemma4} holds. Then 
for any Gibbs measure $\mu_z$, 
\begin{align*}
\mu_z(A_{cont}) &\le \sum_{K=1}^\infty \# \Gamma_K \,(\pi \delta^2 z/4)^{-\lceil c K\rceil} \\
&\le \sum_{K=1}^\infty \left(\frac{(K+1) H}{\epsilon}\right)^2\left(\frac{H}{\epsilon}\right)^{2(K-1)} \left(\frac{\pi \delta^2 z}{4}\right)^{-\lceil c K\rceil}.
\end{align*}
This shows that $\mu_z(A_{cont})\to 0$ as $z\to \infty$, uniformly in $\mu_z$.

Now for any $\omega \in A_{orig}$, $d(0,\Psi(\omega))> \delta + 2r$, and so 
$d(0,\omega) > \delta/2 + 2r$. 
It follows that for any $\omega \in A_{orig}$ and any $x \in B_{\delta/2}(0)$, $\omega \cup x \in \Omega$. 
A simplified version of the proof of Lemma~\ref{lemma3} then implies that 
$\mu_z(A_{orig}) \le (\pi \delta^2 z/4)^{-1}$ for any Gibbs measure $\mu_z$. 
Thus $\mu_z(A_{orig})\to 0$ as $z\to \infty$ uniformly in $\mu_z$, and 
the result follows.
\end{proof}

Below Theorem~\ref{theorem1} is extended to continuous space:

\begin{theorem}
 \label{theorem2}
Let $L > 3r$. Let $A_{inf}$ be the set of all $\omega \in \Omega$ such that $\cup_{x \in \omega}B_{L/2}(x)$ 
has an infinite connected component, $W$, with $d(0,W) \le L/2$. Then $A_{inf} \in \mathcal F$,  
and $\lim_{z\to \infty} \mu_z(A_{inf}) = 1$, uniformly in all Gibbs measures $\mu_z$.
\end{theorem}

\begin{proof}
It is standard to show $A_{inf} \in \mathcal F$, so this 
part of the proof is omitted. To see that $\lim_{z\to \infty} \mu_z(A_{inf}) = 1$, 
choose $\delta \in (0,r/2)$ and $\epsilon \in (0,\delta/2)$ such that $3r+2\delta + 2\epsilon < L$, and 
define $A_{inf}^\Psi$ as in Theorem~\ref{theorem1}. 
Then $A_{inf}^\Psi \subset A_{inf}$ and so $\mu_z(A_{inf}) \ge \mu_z(A_{inf}^\Psi)$. 
The result now follows from Theorem~\ref{theorem1}.
\end{proof}

The main result can now be proved:

\begin{theorem}
 \label{theorem3}
Let $L > 3r$. Let $A$ be the set of all $\omega \in \Omega$ such that $\cup_{x \in \omega}B_{L/2}(x)$ 
has an infinite connected component. Then $A \in \mathcal F$, and for $z$ 
sufficiently large, $\mu_z(A) = 1$ for all Gibbs measures $\mu_z$.
\end{theorem}

\begin{proof}
Proof of measurability is again omitted. 
It is clear that $A$ is in the tail sub-$\sigma$-algebra of $\mathcal F$, 
so $\mu_z(A) = 0$ or $1$ for all extremal Gibbs measures $\mu_z$
(see \cite{Georgii2}, Chapter 7, Theorem 7.7). 
Let $A_{inf}$ be defined as in Theorem~\ref{theorem2}. Since $A_{inf} \subset A$,  
Theorem~\ref{theorem2} implies that $\lim_{z \to \infty} \mu_z(A) = 1$ uniformly in 
all Gibbs measures $\mu_z$. So for $z$ sufficiently large, 
$\mu_z(A) = 1$ for all extremal Gibbs measures $\mu_z$. 
The result now follows from extremal decomposition of Gibbs measures 
(see \cite{Georgii2}, Chapter 7, Theorem 7.26).
\end{proof}

\section{Conclusion}

Percolation of excluded volume has been proved for points in the plane distributed 
according to Gibbs measures with a pure hard core interaction. This model, 
commonly called the {\it hard disk model}, is among the simplest continuum models of particles 
with pair interactions. The proof, which generalizes to 3D, relies on a Peierls-type argument \cite{Griffiths}. 
(The generalization requires a slightly more complicated argument for choosing $u_0$ 
and estimating the number of contours of a given size.) 
A similar result is expected in a hard disk model with an added attraction which 
extends beyond the hard core, though this generalization is not pursued here. 
The hard disk model with attraction is believed to exhibit a gas-liquid phase transition, which 
has been heuristically connected to percolation of excluded volume \cite{Kratky},\cite{Woodcock}. 
(There is no proof in the literature of a gas-liquid transition 
in a continuum model with pair interactions; see, however, \cite{Lebowitz}.) 
To this author's knowledge, there is no previous proof of percolation of 
excluded volume for hard disks 
(or spheres) in the literature. (See \cite{Radin} for a proof in a model with a complicated exclusion.) 
In general, very little is known (or proved) about the qualitative properties 
of the hard disk model at large activity. The result of this paper is of particular interest 
because of the known absence of long range translational order in the model. It remains an open 
question whether percolation occurs for an arbitrarily small connection radius, that is, 
for a connection radius extending just beyond the exclusion radius \cite{Radin}. 

\section{Acknowledgement}

The author wishes to thank C. Radin for useful discussions.

\end{document}